\documentclass[a4paper,UKenglish,cleveref, autoref]{lipics-v2019}
\usepackage{url,graphics,amsmath,amsfonts,amssymb,tikz,verbatim}
\usepackage{algorithm}
\usepackage[noend]{algpseudocode}
\usepackage{mathtools}
\DeclarePairedDelimiter{\ceil}{\lceil}{\rceil}
\newcommand{\HIGH}{\sf HIGH}
\newcommand{\LOW}{\sf LOW}
\newcommand{\N}{\sf N}
\newcommand{\A}{\sf A}
\newcommand{\U}{\sf USED}
\newcommand{\F}{\sf FREE}
\newcommand{\UC}{\sf USED\_A}
\newcommand{\FC}{\sf FREE\_A}
\newcommand{\CT} {\sf COUNT}
\newcommand{\Insert}{\sf Insert}
\newcommand{\Delete}{\sf Delete}
\newcommand{\getColor}{\sf getColor}

\title{Trade-offs in dynamic coloring for bipartite and general graphs} 

\titlerunning{Dynamic coloring for Bipartite and General Graphs}
\author{Manas Jyoti Kashyop, N. S. Narayanaswamy, Meghana Nasre and Sai Mohith Potluri}{Department of Computer Science and Engineering \\ Indian Institute of Technology Madras, Chennai, India}{\{manasjk, swamy, meghana\}@cse.iitm.ac.in, saimohith71@gmail.com}{}{}
\authorrunning{M. J. Kashyop, N. S. Narayanaswamy, M. Nasre, and S. M. Potluri}
\Copyright{Manas Jyoti Kashyop, N. S. Narayanaswamy, Meghana Nasre and Sai Mohith Potluri}

\nolinenumbers

\begin{document}

\maketitle
\begin{abstract}
We present trade-offs involving update time, query time, and number of recolorings in the incremental and fully dynamic settings to maintian  a proper coloring of bipartite graphs and general graphs.  
We present the following trade-off for a bipartite graph with $n$ vertices:
\begin{itemize}
\item For any fully dynamic $2$-coloring algorithm, the maximum of the update time, number of recolorings, and query time is $\Omega(\log n)$. 
\item There exists a deterministic fully dynamic $2$-coloring algorithm with $O(\log^2 n)$ amortized update time, $O(\log n)$ amortized query time, and one recoloring in the worst case.
\item For any incremental $2$-coloring algorithm which explicitly maintains the color of every vertex after each update, the amortized update time and the amortized number of recolorings is $\Omega(\log n)$. Further, for such an algorithm, in the worst case the  update time and the number of recolorings is $\Omega(n)$.
\item There exists a deterministic incremental $2$-coloring algorithm which explicitly maintains the color of every vertex after each update, with amortized $O(\log n)$ update time and amortized $O(\log n)$ many recolorings. Further, for the same algorithm, in the worst case the  update time and the number of recolorings is $O(n)$.
\item There exists a deterministic incremental $(1+2 \log n)$-coloring algorithm which explicitly maintains the color of every vertex after each update, with $O(\alpha(n))$ amortized update time, one recoloring in the worst case and $O(1)$ worst case query time, where $\alpha(n)$ is the inverse Ackermann function. 
\item There exists a deterministic incremental $2$-coloring algorithm which does not maintain color of every vertex after each update, with amortized $O(\alpha(n))$  update time, amortized $O(\alpha(n))$ many recolorings, and amortized $O(\alpha(n))$ query time, where $\alpha(n)$ is the inverse Ackermann function.
\end{itemize}	
For general graphs and graphs of bounded arboricity, we present the following results, where $\Delta$ is the maximum degree of a vertex.
\begin{itemize}
\item There exists a deterministic $(\Delta+1)$-coloring algorithm with $O(\sqrt{m})$ update time, $O(1)$ query time, and one recoloring, all in the worst case. Here $m$ denotes the maximum number of edges throughout the update sequence.
\item There exists a deterministic $(\Delta+1)$-coloring algorithm with
amortized $O(\gamma + \log{n})$ update time, $O(1)$ worst case query time, and one recoloring in the worst case. Here $\gamma$ is the bound on arboricity. 
\end{itemize}
\end{abstract}

\section{Introduction}
\label{sec:Introduction}
Given an undirected graph $G(V, E)$, a proper $c$-coloring of $G$ assigns to each vertex
in $G$ a  color from the set $\{1, 2, \ldots, c\}$ such that no two adjacent vertices
are assigned the same color.
A graph $G$ is $2$-colorable {\em iff} $G$ is bipartite \cite{CLRS}. 
Further, it is well-known that a graph can be vertex colored using at most $\Delta+1$ colors, where $\Delta$ is the maximum degree.  A fully dynamic coloring algorithm supports the $\Insert$$(u,v)$ to insert an edge between vertex $u$ and vertex $v$,
$\Delete$$(u,v)$ to delete the edge $(u,v)$, and
$\getColor$$(u)$ to query the color of vertex $u$.  An incremental coloring algorithm only supports $\Insert$ and $\getColor$ operations. 
The goal is to maintain a proper $c$-coloring of the vertices in $V$ of $G$. $c$  participates in a trade-off with the query time and the update time.   In this work we  also consider the number of recolorings per update as a parameter in a trade-off.   Each recoloring counts a vertex whose color is modified during an update. The number of recolorings per update is a special case of {\em adjustment complexity} \cite{DBLP:AssadiOSS18} which is a bound on the number of vertices or edges affected per update.\par  
We prove our lower bound on the trade-off in the  cell-probe model.  Cell-probe model  has been a standard for proving dynamic data structure lower bounds, as in the works of Fredman and Saks \cite{Fredman:1989:CPC:73007.73040} and P\u{a}tra\c{s}cu and Demaine \cite{PatrascuD06}.  In this model the memory is organized into words and a data structure occupies a certain number of words.  The queries are processed by an online algorithm which may update the data structure. As the query sequence is processed one query after the other, the online algorithm proceeds by reading from or writing words into the data structure.   The running time of the dynamic algorithm is defined to be equal to the number of probes into the data structure.  One natural trade-off is between the size of the data structure and the running time of the dynamic algorithm \cite{Miltersen99cellprobe}.
 The trade-off between the query time $t_q$ (number of times a word is read from the data structure) and the update time $t_u$ (number of times a word is written into the data structure) is also of interest.  P\u{a}tra\c{s}cu has proved lower bounds on the trade-off between $t_u$ and $t_q$, and some of these are proved based on a lower bound on the communication complexity of lopsided (asymmetric) set disjointness \cite{DBLP:journals/corr/abs-1010-3783}.  

\noindent
{\bf {Our results on bipartite graphs}}\\
In Section~\ref{subsec:lowerbounds-bipartite}, we present a lower bound in the cell-probe model on the  trade-off between update time, query time, and 
the number of recolorings for a fully dynamic 2-coloring algorithm.  In the cell-probe model we additionally focus on the words that store the colors assigned to the vertices.   Let $t_R$ be the number of times a word storing the color of a vertex is modified during an update and  $t_D$ = $t_u - t_R$. So, $t_D$ denotes the number of modifications to words which do not store a vertex color. 
\renewcommand{\labelitemi}{$\bullet$}
\begin{itemize}
\item In the cell-probe model of computation, any data structure for fully dynamic 2-coloring of bipartite graph satisfies the following trade-offs : $t_q \log\left(\frac{t_u}{t_q}\right)$ = $\Omega\left(\log n\right)$ and $t_u \log\left(\frac{t_q}{t_u}\right)$ = $\Omega(\log n)$. These bounds hold under amortization, non-deterministic queries, and randomization, and even if the graph is a disjoint union of paths.
Further, $\max \{t_D, t_R, t_q\} = \Omega(\log n)$. In other words, either the query time is $\Omega(\log n)$, or number of recolorings is $\Omega(\log n)$, or $t_D$ is $\Omega(\log n)$ (Theorem~\ref{thm:LowerBound-2coloring}).  
\end{itemize}
 Independent of this work, a similar reduction is used in \cite{DBLP:journals/corr/Henzinger} to show that any data structure for dynamic $\Delta$-colorability testing, where $\Delta$ is the maximum degree in the graph, must perform $\Omega(\log n)$ cell probes.\par
If the color of every vertex is explicitly maintained after every update then the query $\getColor$ takes $O(1)$ time, and we refer to this as the maintenance of an explicit coloring. Consequently, From the above lower bound result, $t_D + t_R$ is $\Omega(\log n)$.  This leads to a natural question of whether there is an
algorithm that takes $\Omega(\log n)$ time for the queries and $o(\log n)$ time for the updates.  We address this question by considering algorithms that do not maintain the color of every vertex after each update. However, such an algorithm computes the required color on a $\getColor$ query. We refer this approach as maintaining an implicit coloring. 
In Section~\ref{sec : Bipartite_Fully_Dynamic_Algorithm} we present a fully dynamic implicit 2-coloring algorithm.
\begin{itemize}
\item There exists a deterministic fully dynamic implicit $2$-coloring algorithm for bipartite graphs with amortized $O(\log^2 n)$ update time, amortized $O(\log n)$ query time, and one recoloring in the worst case (Theorem~\ref{thm:FullyDynamicBipartiteImplicit}).  
\end{itemize}
\noindent
In Section~\ref{subsec:incremental-LB} we explore if the lower bound result in Theorem~\ref{thm:LowerBound-2coloring} can be overcome in the incremental setting for maintaining an explicit 2-coloring.   We show that deterministic incremental explicit 2-coloring algorithms  have challenging update sequences generated by an adaptive adversary.
\begin{itemize}
\item For any deterministic incremental explicit 2-coloring algorithm for bipartite graph, amortized update time is $\Omega(\log n)$ and worst case update time is $\Omega(n)$. Further, for such an algorithm, amortized number of recolorings is $\Omega(\log n)$ and worst case number of recolorings is $\Omega(n)$ (Lemma~\ref{lem:lowerbound-incremental}).  
\end{itemize}
We present an upper bound  which matches the lower bound in Lemma~\ref{lem:lowerbound-incremental} by using the classical decremental connectivity algorithm due to Even and Shiloach \cite{Shiloach:1981:OEP:322234.322235}.
\begin{itemize}
\item There exists a deterministic incremental explicit $2$-coloring algorithm with amortized $O(\log n)$ update time and amortized $O(\log n)$ many recolorings. Further, for the same algorithm, in the worst case, the  update time and the number of recolorings is $O(n)$.  
\end{itemize} 
In Section~\ref{subsec:Incremental-Explicit-Coloring}, we next consider incremental explicit coloring of bipartite graphs using  $(1+ 2\log n)$ colors and show that it is possible to have significantly smaller query and update time, and number of recolorings in comparison to what is achieved by maintaining a 2-coloring.
\begin{itemize}
\item There exists a deterministic {\em explicit} $(1+2\log n)$-coloring algorithm for bipartite graphs with amortized $O(\alpha(n))$ update time, $O(1)$ worst case query time, and one recoloring in the worst case, where $\alpha(n)$ is the inverse Ackermann function (Theorem~\ref{thm:IncrementalBipartiteExplicit}).
\end{itemize}
In Section~\ref{subsec:Incremental-Implicit-Coloring}, we present the following result:
\begin{itemize}
\item There exists a  deterministic implicit $2$-coloring
incremental algorithm for bipartite graphs with amortized $O(\alpha(n))$ update time, query time, and number of recolorings, where $\alpha(n)$ is the inverse Ackermann function (Theorem~\ref{thm:IncrementalBipartiteImplicit}). 
\end{itemize}
\noindent 
{\bf {Our results on general graphs and bounded arboricity graphs}}\\
In Section~\ref{sec: GeneralGraph}, we present our result for general graphs in the fully dynamic setting.     
\begin{itemize}
\item There exists a deterministic $(\Delta+1)$-coloring algorithm with
$O(\sqrt{m})$ update time, $O(1)$ query time, and one recoloring in the worst case, where $\Delta$ is the maximum degree of a vertex and $m$ is the maximum number of edges in the graph at any point during the update sequence (Theorem~\ref{thm:FullyDynamicGeneralGraph}).
\end{itemize}
The final result, presented in Section \ref{sec:Low-Arboricity-Graphs}, is for the special case of bounded arboricity graphs.
\begin{itemize}
\item For graphs with arboricity bounded by $\gamma$, there exists a fully dynamic deterministic 
($\Delta$+1)-coloring algorithm with $O(\gamma + \log n)$ amortized update time, $O(1)$ worst case query time, and one recoloring in the worst case, where $n$ is the number of vertices in the graph (Theorem~\ref{thm:FullyDynamicLowArboricity}).
\end{itemize}
\noindent
The algorithms mentioned above maintain a $\Delta + 1$ coloring 
as in the current best randomized algorithms in \cite{DBLP:journals/corr/sayan} and \cite{DBLP:journals/corr/Henzinger}. 
These algorithms are known to perform $\Omega(\log n)$  and  $\Omega(n)$ recolorings, respectively.   On the other hand, our deterministic algorithm performs only a constant number of  recolorings per update.  This reduction in the number of recoloring comes at the cost of the update time which is $O(\sqrt{m})$ in the worst-case.  On the other hand, being  deterministic, our algorithm (Theorem~\ref{thm:FullyDynamicGeneralGraph}) works even against an adaptive adversary, whereas the randomized algorithms in \cite{DBLP:journals/corr/sayan}, and \cite{DBLP:journals/corr/Henzinger} assume that adversary is non-adaptive and oblivious. This is similar to the case of the fully dynamic maximal matching problem where the best randomized algorithm takes $O(1)$ update time \cite{DBLP:conf/focs/Solomon16} while the best known deterministic algorithm takes $O(\sqrt{m})$ update time \cite{Neiman}. In the fully dynamic setting, it remains open as to whether there is a deterministic $(\Delta + c)$-coloring algorithm with update time $o(\sqrt{m})$ for some constant $c >1$. In Table~\ref{tab:General} we tabulate the current results \cite{DBLP:journals/corr/sayan,DBLP:journals/corr/Henzinger} along with our results on the different parameters in the trade-offs we consider.\\
\begin{table}[h!]
	\footnotesize	
	\caption{Note that our deterministic fully dynamic 2-coloring algorithm for bipartite graphs $\max\{t_D, t_R, t_q\} = O(\log^2 n)$, and this is far from the  $\Omega(\log n)$ lower bound in Theorem~\ref{thm:LowerBound-2coloring}.}   
	\centering	
	\begin{tabular}{|m{2cm} | m{2.2cm} | m{1.5cm}| m{2cm} | m{2cm}| m{2cm}|} 
		\hline \vspace{2mm}
		Graph  & Type & Number & Update time & Query time & Number of \\
		&	& of colors &  &  & recolorings in the worst case \\ 
		\hline \vspace{2mm}
		Bipartite	& Deterministic  & $2$ & $O(\log^2 n)$ & $O(\log n)$ & $1$  \\
		(Theorem~\ref{thm:FullyDynamicBipartiteImplicit}) & fully dynamic & & amortized & amortized & \\
		\hline \vspace{2mm}
		Bipartite	& Deterministic  & $1+2\log n$ & $O(\alpha(n))$ & $O(1)$ & $1$  \\
		(Theorem~\ref{thm:IncrementalBipartiteExplicit})& incremental & & amortized & worst case & \\
		\hline \vspace{2mm}
		Bipartite	& Deterministic  & $2$ & $O(\alpha(n))$ & $O(\alpha(n))$ & $\Omega(\log n)$\\
		(Theorem~\ref{thm:IncrementalBipartiteImplicit})& incremental & & amortized & amortized & \\
		\hline \vspace{2mm}
		General	& Deterministic  & $\Delta + 1$ & $O(\sqrt{m})$ & $O(1)$ & $1$  \\
		(Theorem~\ref{thm:FullyDynamicGeneralGraph})& fully dynamic & & worst case & worst case & \\
		\hline \vspace{2mm}
		$\gamma-$arboricity & Deterministic  & $\Delta + 1$ & $O(\gamma + \log n)$ & $O(1)$ & $1$  \\
		(Theorem~\ref{thm:FullyDynamicLowArboricity}) & fully dynamic & & amortized & worst case& \\
		\hline \vspace{2mm}	
		General \cite{DBLP:journals/corr/sayan} &	Randomized  & $\Delta + 1$ & $O(1)$  & $O(1)$ & $\Omega(\log n)$ \\
		&fully dynamic & & amortized & worst case & \\ 
		\hline \vspace{2mm}
		General \cite{DBLP:journals/corr/Henzinger} & Randomized & $\Delta + 1$ & $O(1)$  & $O(1)$ & $\Omega(n)$ \\
		& fully dynamic & & amortized & worst case & \\
		\hline
	\end{tabular}
	\label{tab:General}	
\end{table}
\textbf{Related work: } Bhattacharya~et~al.\cite{BCHN18} presented a randomized $(\Delta+1)$-coloring 
algorithm with expected amortized $O(\log\Delta)$ update time. Very recently, Bhattacharya et~al.\cite{DBLP:journals/corr/sayan} and independently Henzinger et~al.\cite{DBLP:journals/corr/Henzinger} presented randomized algorithms for $(\Delta+1)$-coloring in constant amortized update time. In \cite{BCHN18}, Bhattacharya~et~al.  also presented a deterministic algorithm with $O(\text{polylog} \Delta)$ amortized update time, but using $(\Delta+ o(\Delta))$ colors. 
In a different work, Barba~et~al.~\cite{Barba} studied 
the trade-off between the number of colors used and the number of recolorings per update. 
For any $d > 0$, they presented two algorithms, one using $O(Cdn^{\frac{1}{d}})$ colors and 
performing $O(d)$ recolorings and the other using $O(Cd)$ colors and performing $O(dn^{\frac{1}{d}})$ recolorings, where $C$ is the chromatic number. 
However, both these algorithms require exponential running time. Solomon~et~al.~\cite{SolomonW18} presented an algorithm which uses $\tilde{O}(\frac{C}{d} \log^{2}n)$ colors and performs $O(d)$ recolorings per update, where $\tilde{O}$ suppresses polyloglog$(n)$ factors.
\section{Fully dynamic coloring for bipartite graphs}
\label{sec:Incremental-Bipartite}
In this section we introduce the number of recolorings into the well-known, due to P\u{a}tra\c{s}cu et al. \cite{PatrascuD06}, trade-off between update time and query time in the cell-probe model. We then present our fully dynamic algorithm that maintains an implicit 2-coloring with one recoloring.
\subsection{Lower bound on full dynamic 2-coloring of bipartite graphs}
\label{subsec:lowerbounds-bipartite}
Barba et al. \cite{Barba} showed that for any $c \geq 2$, any fully dynamic algorithm that maintains a $c$-coloring of a 2-colorable graph on $n$ vertices must recolor at least $\Omega(n^{\frac{2}{c(c-1)}})$ vertices per update. Their result immediately implies that if we use only two colors for bipartite graphs, then the worst case update time for any fully dynamic algorithm with $O(1)$ query time is $\Omega(n)$. We present a trade-off between update and query time, and the number of recolorings for a fully dynamic 2-coloring algorithm. Recall the definition of $t_u, t_q, t_D$, and $t_R$ from the introduction.  
\begin{theorem}
	\label{thm:LowerBound-2coloring}	
	In the cell-probe model of computation, any data structure for fully dynamic 2-coloring of bipartite graph satisfies the following trade-offs : $t_q \log\left(\frac{t_u}{t_q}\right)$ = $\Omega\left(\log n\right)$ and $t_u \log\left(\frac{t_q}{t_u}\right)$ = $\Omega(\log n)$. These bounds hold under amortization, non-deterministic queries, and randomization, and even if the graph is a disjoint union of paths.
	Further, $\max \{t_D, t_R, t_q\} = \Omega(\log n)$. In other words, either query is $\Omega(\log n)$, or number of recolorings is $\Omega(\log n)$, or $t_D$ is $\Omega(\log n)$.   
\end{theorem}
\begin{proof}
Our proof is by a reduction from fully dynamic connectivity for bipartite graphs to fully dynamic  2-coloring for bipartite graphs.    Let $\mathcal{A}$ be a fully dynamic algorithm for 2-coloring of bipartite graphs. $\mathcal{A}$ reports violation of bipartiteness if insertion of an edge creates an odd length cycle and supports the following operations: $\Insert$$(u,v)$ to insert an edge between $u$ and $v$, $\Delete$$(u,v)$ to delete the edge between $u$ and $v$, and $\getColor$$(u)$ to report the color of vertex $u$. In the reduction, the response to a  connectivity query characterized by  the response to a constant number of insert,  query, and delete updates presented $\mathcal{A}$. As a consequence of this reduction and the known lower bound for fully dynamic connectivity for bipartite graphs from P\u{a}tra\c{s}cu et al. \cite{PatrascuD06}  we show the lower bound in the Theorem~\ref{thm:LowerBound-2coloring} follows.\par
We now present our reduction by using $\mathcal{A}$ to solve fully dynamic connectivity in bipartite graphs. The vertex set of the bipartite graph consists of two new additional vertices $a$ and $b$.
An edge $(u,v)$ is inserted or deleted using the operation $\Insert$$(u,v)$ or $\Delete$$(u,v)$, respectively.  A connectivity query is of the form : \textit{are $u$ and $v$ connected}? and it is reduced to a sequence of operations in $\mathcal{A}$ as follows: first obtain  the color of $u$ and $v$ using $\getColor$$(u)$ and $\getColor$$(v)$.\\
\textbf{Case $1$:} $u$ and $v$ are of same color.  The following sequence of updates are presented to $\mathcal{A}$: $\Insert$$(u,a)$, $\Insert$$(v,b)$, $\Insert$$(a,b)$. 
If $\Insert$$(a,b)$ violates bipartiteness, then report that $u$ and $v$ are connected. Otherwise, report that $u$ and $v$ are not connected. \\ 
\textbf{Case $2$:} $u$ and $v$ are of different color.  The following sequence of updates are presented to $\mathcal{A}$: $\Insert$$(u,a)$, $\Insert$$(v,a)$.   If $\Insert$$(v,a)$ violates bipartiteness, then report that $u$ and $v$ are connected. Otherwise, report that $u$ and $v$ are not connected. \\
 Finally, in either case  $\mathcal{A}$ is presented with updates which delete the edges inserted which are incident on $a$ or $b$.  Further, in either case the decision on whether $u$ and $v$ are connected is correct, and this follows from the fact that a connected bipartite graph has a unique bipartition.
 
 Due to the reduction, it follows from the known lower bound for fully dynamic connectivity for bipartite graphs from P\u{a}tra\c{s}cu et al. \cite{PatrascuD06} that
 $\max \{t_u t_q\} = \Omega(\log n)$.   By definition, $t_R$ is the time spent by cell-probes that changes the color value of vertices, and $t_D = t_u - t_R$.  Therefore, it follows that $\max \{t_D, t_R, t_q\} = \Omega(\log n)$. Hence the theorem.
%
\end{proof} 
In case of explicit 2-coloring of bipartite graphs $t_q = O(1)$.  Therefore, it follows from Theorem \ref{thm:LowerBound-2coloring}, either $t_D$ = $\Omega(\log n)$ or $t_R$ = $\Omega(\log n)$.  In other words, if an algorithm maintains an explicit 2-coloring at the end of each update, then during the update it must recolor many vertices,  and if it does not recolor many vertices, it must perform many probes into the data structure.
\vspace{2cm}
\subsection{Fully dynamic  deterministic implicit 2-coloring}
\label{sec : Bipartite_Fully_Dynamic_Algorithm}
\noindent
Our  algorithm for implicit 2-coloring uses the fully dynamic connectivity algorithm due to Holm et al. \cite{HLT} and the link-cut tree data structure \cite{ST} to maintain an additional representation of a spanning forest.\\
\textbf{Implicit coloring: }To maintain an implicit coloring, we select an arbitrary vertex $r$ in a connected component, which we refer to as the root of the component.  Since a connected bipartite graph has a unique bipartition, fixing the color of the root immediately defines the color of all the other vertices in the component. Thus, on an update only the color of the root of the appropriate component is changed. The color of each vertex is thus defined by the the parity of the path length from the vertex to the root of its component: it is same as the color of the root if the parity is even, and the other color if the parity is odd.  A dynamic algorithm for 2-coloring can be designed by using algorithms for connectivity and computing the parity of the length of a path. 
Further, an edge insertion between two vertices of the same color in a connected component is forbidden in a bipartite graph.   Thus, during an insert, 
the algorithm due to Holm et al.  \cite{HLT} is also used to check whether two vertices of the same color belong to the same connected component.  The parity of a path length is computed using link-cut trees \cite{ST} \par
The  idea for the algorithm due to Holm et al. \cite{HLT} is to store a spanning tree for every connected component of the input graph $G$. The input graph $G$ is hierarchically partitioned into $\log n$ levels, where $n$ is the total number of vertices, by assigning levels to the edges. The level of an edge is an integer in the interval $[0,\log n]$.  
 The algorithm maintains the invariant that the level of an edge can only decrease over time and this is crucially used to achieve the amortized bounds in the algorithm. For each $0 \leq i \leq \log n$,  $G_i$ is the subgraph of $G$ composed of edges of level at most $i$. Therefore, $G_{\log n}$ is $G$.  Further, for each $0 \leq i \leq \log n$, $F_i$ is the spanning forest of $G_i$ in which each edge has a weight which is the level of the edge. Therefore, $F_{\log n}$ is the spanning forest for $G$ with one tree for each connected component of $G$.  Every spanning forest is maintained using the  Euler-Tour(ET) tree data structure \cite{HK}.
 The data structure due to Holm et al. supports the following operations:\\
\textit{Connected(u,v)} : Reports if $u$ and $v$ are in the same connected component in $\mathcal{O}(\log n)$ time. \\
\textit{Conn-Insert(u,v)} : Insert an edge between $u$ and $v$. This operation takes $\mathcal{O}(\log n)$ time. \\
\textit{Conn-Delete(u,v)} : Delete the edge $(u,v)$. This operation takes amortized $\mathcal{O}(\log^2 n)$ time. We use the data structure with a modification to the Delete in Holm et al. \cite{HK} by ensuring that it returns a value. If $(u,v)$ is a non-tree edge, then a 0 is returned. If $(u,v)$ is a tree edge that has a replacement edge as computed by Holm et al. algorithm, that edge is returned. Otherwise -1 is returned.
However,  ET-trees do not provide for an efficient way to extract the  properties of a path between two vertices.  This is addressed by additionally maintaining a copy of $F_{\log n}$ using link-cut trees \cite{ST}. The link-cut tree supports the following operations in amortized $\mathcal{O}(\log n)$ time.\\
\textit{Find-root(v) : } Returns root of the tree which contains vertex $v$.\\ 
\textit{Link(u,v) :} Inserts the edge $(u,v)$ into the link-cut tree and makes $u$  a new child of $v$.\\
\textit{Cut(v):} Deletes the edge between $v$ and parent of $v$ where $v$ is not a root.\\
\textit{Path-Length(v) :} Returns the length of the path from the root of the tree to $v$.\\ 
\textbf{Implicit coloring using Holm et al. and link-cut trees:} \\
Our data structure for maintaining a 2-coloring has two colors in its palette : {\sf TRUE} and {\sf FALSE}. It is initialized with an empty graph on $n$ vertices with {\sf TRUE} as the vertex color.  Then, the data structures of Holm et al. and the link-cut tree are initialized.  
Every vertex stores an attribute indicating its color. The query $\getColor$$(v)$ is obtained by first finding
 the root, say $r$, of the tree containing $v$ using \textit{Find-root}$(v)$. The length of the path from $v$ to $r$ is obtained using a query to \textit{Path-Length}$(v)$. 
 If the length is even, then we output the  color of $r$ as the color of $v$, otherwise we output the negation of the color of $r$. This operation takes the same time as \textit{Find-root} and \textit{Path-Length} and hence contributes $\mathcal{O}(\log n)$ to the amortized running time. The updates $\Insert$$(u,v)$ and $\Delete$$(u,v)$ are described below.\\
\textbf{Insert(u,v):} On the insert of an edge $(u,v)$, the color of $u$ and $v$ is computed using $\getColor$$(u)$ and $\getColor$$(v)$, respectively. We have the following cases:
\begin{enumerate}
\item Color of $u$ is different from color of $v$, no recoloring is required.  Call \textit{Connected$(u,v)$}.
\begin{enumerate}
\item Case when $u$ and $v$ belong to the same component.  Call \textit{Conn-Insert(u,v)} to insert the edge $(u,v)$ into the Holm et al. data structure.
\item Case when $u$ and $v$ belong to different components. Call \textit{Conn-Insert(u,v)} to insert the edge into the Holm et al. data structure. This is followed by a call to \textit{Link}$(u,v)$ to update the link-cut tree, since the edge will be added to the forest $F_{\log n}$.
\end{enumerate}  
\item Color of $u$ is same as color of $v$, this necessitates a possible recoloring.   Call \textit{Connected}$(u,v)$.
\begin{enumerate}
\item Case $u$ and $v$ belong to the same component. Then, edge $(u,v)$ is ignored and no data structure is udpated.
\item Case $u$ and $v$ belong to different components.  Call \textit{Conn-Insert(u,v)} to insert the edge into the Holm et al. data structure. This is followed by a call to \textit{Link}$(u,v)$ to update the link-cut tree, since the edge will be added to the forest $F_{\log n}$, thus reducing the number of components. 
\end{enumerate}
\end{enumerate}
\textbf{Delete(u,v):}
While deletion does not necessitate any color change, the connectivity information must be appropriately updated.  This is done by \textit{Conn-Delete(u,v)}, and the following exhaustive cases are considered based on the return value of \textit{Conn-Delete(u,v)}.\\  
\begin{enumerate}
\item Case edge $(u,v)$ is a non-tree edge. The forest $F_{\log n}$ does not change, and thus the link-cut tree also need not change.
\item Case edge $(u,v)$ is a tree edge. Wlog let  $u$ be parent of $v$ in $F_{\log n}$ prior to the delete.  If $(u,v)$ does not have a replacement edge, then the forest $F_{\log n}$ has one more component after the delete,  then the link-cut tree is updated by a call to \textit{Cut}$(v)$.  On the other hand, if the replacement edge is $(x,y)$, then call \textit{Cut}$(v)$ followed by a call to \textit{Link}$(x,y)$.
\end{enumerate}
\textbf{Analysis :} Since the colors of all the vertices were initialized to {\sc True}, throughout the execution each vertex has some color associated with it.  In a bipartite connected graph the color of one vertex fixes the colors of all the components.  Therefore, $\getColor$ will return the correct color of each vertex in a component relative to the color of the root vertex of the component in the link-cut tree.  Therefore, the data structure maintains an implicit 2-coloring in the fully dynamic setting.  Other than the operations supported by Holm et al. \cite{HK}, the link-cut tree and the $\getColor$ operation, every other operation  in our algorithm takes constant time. In any update step, there are at most a constant many calls to the Holm et al. data structure, the link-cut tree, and  $\getColor$.    Therefore, our algorithm supports updates in amortized $O(\log^2 n)$ time and reports the color of a vertex in amortized $O(\log n)$ time. Thus, we have the following Theorem. 
\begin{theorem}
	\label{thm:FullyDynamicBipartiteImplicit}
	Starting with an empty bipartite graph with $n$ vertices, an implicit 2-coloring can be maintained deterministically over any sequence of insertion and deletion in amortized $O(\log^2 n)$ update time, amortized $O(\log n)$ query time, and one recoloring in the worst case.
\end{theorem}

\vspace{1cm}

\section{Trade-offs in incremental coloring for bipartite graphs}
\label{subsec:Bipartite-incremental}
We start this section by presenting our results on incremental explicit 2-coloring.
\vspace{-2mm}
\subsection{Lower bound on incremental explicit 2-coloring}
\label{subsec:incremental-LB}
\begin{lemma}
\label{lem:lowerbound-incremental}	
For any deterministic incremental algorithm that maintains an explicit 2-coloring of a bipartite graph, amortized update time is  $\Omega(\log n)$ and worst case update time is $\Omega(n)$. Further, for such an algorithm, amortized number of recolorings is $\Omega(\log n)$ and worst case number of recolorings is $\Omega(n)$.
\end{lemma}
\begin{proof}
We prove the lemma by constructing an update sequence such that any deterministic incremental algorithm which maintains an explicit 2-coloring will have the claimed update and query time, and  number of recolorings. 
The graph is initialized to be the $n$ vertex graph with no edges, and the deterministic incremental algorithm is assumed to maintain an explicit 2-coloring.  They key observation is that in any 2-coloring, at least  $\frac{n}{2}$ vertices will get the same color.\\
The update sequence is described as follows: an adaptive adversary, who can execute the incremental algorithm, always presents an edge update that is between two vertices of the same color, and its addition merges two paths into a longer path. Further, the edge update is also chosen to guarantee that the length of any two disjoint paths are within a factor of 2 of each other.
Since the algorithms maintains an explicit 2-coloring of a forest of paths, all the vertices in one of the paths must be recolored. 
This ensures that the total update time and total number of recolorings over $n-1$ inserts is $\Omega(n\log n)$. Further, as the paths become longer,   the worst case update time and the number of recolorings is $\Omega(n)$.  Hence the lemma.\\
\end{proof}
\noindent
\textbf{An optimal incremental explicit 2-coloring algorithm: }
We present an explicit 2-coloring deterministic incremental algorithm that takes $O(m \log n)$ time for $m$ inserts and hence optimal in presence of the lower bound in Lemma~\ref{lem:lowerbound-incremental}. The algorithm maintains a 2-coloring using an idea similar to the decremental algorithm for connectivity due to Even and Shiloach \cite{Shiloach:1981:OEP:322234.322235}. It works as follows : We maintain the color for every vertex. If an edge insertion between two vertices of same color merges two connected component then we update the color for the vertices in the smaller component. For a vertex belonging to a smaller component, the component size it is present in has  at least doubled after merging. Therefore, a vertex can participate in at most $\log n$ merges where it is part of the smaller component.  Therefore, the total update time is $O(m \log n)$.
\subsection{Incremental  deterministic explicit ($1+2\log n$)-coloring}
\label{subsec:Incremental-Explicit-Coloring}
In this section, we present a deterministic near-constant update time incremental algorithm for explicit coloring of bipartite graphs using at most $1+2\log n$ colors which performs at most $1$ recoloring per update.  The analysis of the incremental algorithm that we present has some common features with the analysis of the online vertex coloring algorithm \cite{Lovsz}.  However, the scenario in the online vertex coloring is that the vertices are to be colored once in the order presented to the algorithm, and this is different from the incremental setting, where edges are presented one after the other, and vertex colors may be changed to guarantee bounds on the number colors used, number of recolorings, query time, and update time. The algorithm presented here repeatedly selects  an appropriate color for one of the two vertices during an edge insert, and the end ensures that the coloring is proper.  Thus, it performs at most $1$ recoloring per update. \\
\textbf{Algorithm: }
The colors are numbered from $1$ to $1+2\log n$.  
Initially, the graph is empty and  all the $n$ vertices are assigned color $1$. An edge insertion leads to recoloring of a vertex only if both endpoints of the edge are of the same color. Otherwise, no recoloring is required. When an edge is inserted between two vertices of the same color, we arbitrarily choose one of the vertices for recoloring. 
For the recoloring, consider the bipartition, into two independent sets, of the connected component containing the chosen vertex.  Compute the smallest color not used by any vertex in the partition not containing this vertex.   
This color is used to recolor the vertex. Lemma~\ref{thm:2logncolors} bounds the number of colors used by our algorithm. 
\begin{lemma}
\label{thm:2logncolors}
If a connected component in the graph has a vertex with color $t$, then the component has at least $2^{\left \lfloor \frac{t}{2} \right \rfloor}$ vertices. Therefore, total number of colors used by our algorithm is at most $1+2\log n$. 
\end{lemma}
\begin{proof}
	The claim is trivially true for a vertex of color $1$. We proceed by induction on the color $t$. Let $P_1$ and $P_2$ be the two independent sets in the bipartition of the connected component containing $v$, and $P_1$ contains $v$.
	Let the color of $v$ be $t+2$. Since we are in an incremental setting, components can only merge with updates. Given that the graph is bipartite, if a vertex is in $P_2$
	it continues to remain so after all subsequent updates. When $v$ was colored $t+2$, $P_2$ must have had at least one vertex of each color from $1$ to $t+1$ and no vertex of color $t+2$. Let $x$ and $y$ be vertices in $P_2$ colored $t$ and $t+1$, respectively. Similarly, when $y$ was colored $t+1$, $P_1$ must have had a vertex $u$ of color $t$. We have the following cases:\\
	\textbf{case $1$:} Color of $u$ is $t$ at the time of assigning color $t+2$ to $v$. This means that both $P_1$ and $P_2$ have a vertex of color $t$. This could not have happened if both $u$ and $x$ belonged to the same component the last time either of them were recolored. This implies that two components, each of them with a vertex of color $t$, merged to form the current component.\\
	\textbf{case $2$:} Color of $u$ is not $t$ at the time of assigning color $t+2$ to $v$. The only way the color of $u$ could have changed is if there was an edge inserted which is incident on $u$ and  another vertex of color $t$, say $z$. This insertion could have led to the merging of two components, both with a vertex of color $t$. If not, both these vertices are in the same connected component before the insertion of $(u,z)$ and they were already both colored $t$. Since the graph is bipartite, $u$ and $z$ belong to different partitions of the connected component containing them prior to the insertion of $\{u,z\}$. Like the argument in case $1$, we observe that the last time either of $u$ or $z$ was recolored, they must have been in two different connected components.  Subsequent to the  recoloring after which both $u$ and $z$ had the color $t$, the two components merged to form the current component.\par
	From the induction hypothesis we get that in both cases the current component has at least $2^{\left \lfloor \frac{t}{2} \right \rfloor} + 2^{\left \lfloor \frac{t}{2} \right \rfloor} = 2^{\left \lfloor \frac{t+2}{2} \right \rfloor}$ vertices. Hence, by the principle of mathematical induction, our claim is true for all natural numbers $t$. Therefore, total number of colors used by our algorithm is at most $1 + 2 \log n$.
\end{proof}
\textbf{Implementation:}
We use the disjoint-set data structure from the Union-Find algorithm. In addition to $\textit{parent}$ and $\textit{size}$, here, each vertex has an attribute called $\textit{flag}$. It is a boolean attribute. For a vertex $x$, $x.flag = {\sf TRUE}$ iff $x$ and $x.parent$ are in the same bipartition in G. Further, each vertex also has the attributes $color$, two binary words $w1$ and $w2$ of $1 + 2\log n$ bits where the bits are indexed by the integers $0$ to $2\log n$ from right to left, respectively.\\ 
\begin{minipage}{0.5\textwidth}	
\begin{algorithm}[H]
\footnotesize	
	\caption{MAKESET-EXPLICIT initializes the data structures}
	\begin{algorithmic}[1]
		\Function{MAKESET-EXPLICIT}{$x$}
		\If{$x$ not already present:}
		\State{add $x$ to the disjoint set tree}
		\State{$x.parent = x$}
		\State{$x.flag = {\sf TRUE}$}
		\State{$x.size = 1$}
		\State{$x.color = 1$}
		\State{$x.w1 = 0111\ldots$ ($1+2\log n$ bits)}
		\State{$x.w2 = 1111\ldots$ ($1+2\log n$ bits)}
		\EndIf
		\EndFunction
	\end{algorithmic}
	\label{algo: MAKESET-EXPLICIT(x)}
\end{algorithm}
\vspace{1mm}	
\end{minipage}
\hspace{2mm}
\begin{minipage}{0.5\textwidth}
\vspace{3mm}		
	\begin{algorithm}[H]
		\footnotesize
		\caption{FIND computes the root of the component containing vertex $x$}
		\begin{algorithmic}[1]
			\Function{FIND}{$x$}
			\If{$x.parent \neq x$}
			\State{$(root, flag) = FIND(x.parent)$}
			\State{$x.parent = root$}
			\State{$x.flag = flag$ XNOR $x.flag$}
			\EndIf
			\State{RETURN $(x.parent,x.flag)$}
			\EndFunction
		\end{algorithmic}
		\label{algo: FIND(x)}
	\end{algorithm}
\end{minipage}
Procedure MAKESET-EXPLICIT
is called on each vertex (Algorithm~\ref{algo: MAKESET-EXPLICIT(x)}) . This is to initialize the vertices and their attributes. Each vertex is the root of its singleton component, its flag is set to ${\sf TRUE}$, and size to $1$.
The $color$ attribute is initialized to $1$, and as the algorithm progresses the $color$ of each vertex is a proper coloring and it is explicitly maintained. Like $size$, $w1$ and $w2$ are only maintained by the root of the component.   After an update step if a vertex is not  the root of a  set, then $w1$ and $w2$  maintained for the vertex are not used subsequently.  For each vertex $w1$ and $w2$ are status words which indicate which colors are unused (a bit is set if it's index is unused in the corresponding part)  in the two parts of the connected component associated with the set for which the vertex is the root.The word $w1$ is associated with the part which contains the root.\par	
On the insertion of an edge $(x,y)$, procedure UNION-EXPLICIT(Algorithm~\ref{algo: UNION-EXPLICIT(x,y)}) is called.  
Within UNION-EXPLICIT, procedure FIND (Algorithm~\ref{algo: FIND(x)}) is called on vertices $x$ and $y$ to find the roots of their respective components. If they have the same root, then they both belong to the same component. If both have same flag also, then they belong to the same partition which violates bipartiteness of the graph. Otherwise, it is a valid insertion. If both have different roots, the two components are merged by taking the union of the two corresponding sets, and making one of the two roots as the root of the resulting set.  Further,  $w1$ and $w2$ associated wtih the root of the new component is obtained by performing a bitwise AND of the appropriate words associated with the pervious two roots.  In particular, the $w1$ ($w2$) of the root is obtained by taking the bitwise AND of its $w1$ ($w2$) before the union with the $w2$  ($w1$) of the other vertex before the union.  Once maintaining the component information is done we check the colors of $x$ and $y$. If they are different nothing needs to be done. If they are same, one of them needs to be recolored. We recolor the vertex in the same part of the bipartition as the root, which is the part $P_1$, and this is found by computing the flag of the vertex.  The least unused color in $P_2$ is the maximum index of a bit which is a 1 in $w2$ associated with the root. This can be found in constant time using the {\sf CLZ}(count leading zeros) hardware instruction. We use this color to recolor the vertex and then ensure that this color is set to 0  in $w1$ to indicate that it is used to color a vertex in $P_1$. The running time of this algorithm is same as the standard Union-Find algorithm as maintaining the additional attributes and the word operations are all $O(1)$ operations. Hence, the amortized running time per operation is bounded by the inverse Ackermann function which is near constant. We get Theorem~\ref{thm:IncrementalBipartiteExplicit}.		
\begin{theorem}
	\label{thm:IncrementalBipartiteExplicit}
	Starting with an empty bipartite graph with $n$ vertices, an explicit $(1+2\log n)$-coloring, performing one recoloring in the worst case, can be maintained deterministically over any sequence of insertions in amortized $O(\alpha(n))$ update time and worst case $O(1)$ query time, where $\alpha(n)$ is the inverse Ackermann function. 
\end{theorem}
\begin{algorithm}[H]	
\footnotesize	
	\caption{UNION-EXPLICIT handles insertion of edge $(x,y)$}
	\begin{algorithmic}[1]
		\Function{UNION-EXPLICIT}{$x,y$}
		\State{$(xRoot, xFlag) = FIND(x)$}
		\State{$(yRoot, yFlag) = FIND(y)$}
		\If{$xRoot == yRoot$} 
		\If{$xFlag == yFlag$}
		\State{ERROR, INVALID EDGE INSERTION}
		\Else
		\State{RETURN} 
		\EndIf
		\Else 
		\If{$xRoot.size < yRoot.size$}
		\State{swap($xRoot$, $yRoot$)}
		\State{swap($xFlag$, $yFlag$)}
		\EndIf
		\State{$yRoot.parent = xRoot$}
		\State{$xRoot.size = xRoot.size + yRoot.size$}
		\State{$yRoot.flag = xFlag$ XOR $yFlag$}
		\If{$yRoot.flag$}
		\State{$xRoot.w1 = xRoot.w1$ AND $yRoot.w1$}
		\State{$xRoot.w2 = xRoot.w2$ AND $yRoot.w2$}
		\Else
		\State{$xRoot.w1 = xRoot.w1$ AND $yRoot.w2$}
		\State{$xRoot.w2 = xRoot.w2$ AND $yRoot.w1$}
		\EndIf
		\EndIf
		\If{$x.color == y.color$}
		\If{$x.flag$}
		\State{$x.color$ = CLZ$(xRoot.w2) + 1$}
		\State{$xRoot.w1[x.color] = 0$}
		\Else
		\State{$y.color$ = CLZ$(xRoot.w2) + 1$}
		\State{$xRoot.w1[y.color] = 0$}
		\EndIf
		\EndIf
		\EndFunction
	\end{algorithmic}
	\label{algo: UNION-EXPLICIT(x,y)}	
\end{algorithm}
\vspace{-4mm}
\noindent

\subsection{Incremental deterministic implicit 2-coloring}
\label{subsec:Incremental-Implicit-Coloring}
In this section, we present a deterministic near-constant update time incremental algorithm for implicit 2-coloring of bipartite graphs.\\ 
\textbf{Data Structures: }
We use the disjoint-set data structure from the Union-Find algorithm. In addition to $\textit{parent}$ and $\textit{size}$, here, each element also has an attribute called $\textit{flag}$. It is a boolean attribute. For a vertex $x$, $x.flag = {\sf TRUE}$ iff $x$ and $x.parent$ are in the same bipartition in G. In other words, $x.flag = {\sf TRUE}$ iff $x$ and its parent have the same color. The flag attribute can be thought of as the relative coloring between a vertex and its parent. So, the flag for a vertex changes only when its parent changes. We have two colors in our palette: ${\sf TRUE}$ and ${\sf FALSE}$. The flag and the color of the root are always ${\sf TRUE}$. This means that if $x.parent$ is the root, then $x.flag$ represents its color.\\
\begin{minipage}{0.5\textwidth}
	\textbf{Preprocessing:}  MAKESET-IMPLICIT procedure (Algorithm~\ref{algo: MAKESET-IMPLICIT(x)}) is called on on all $n$ vertices in the graph. This is to initialize the vertices and their attributes. Each vertex is the root of its singleton component and its flag is set to ${\sf TRUE}$. The size attribute of a vertex is relevant only if the vertex is a root.
	\vspace{0.08in}
\end{minipage}
\hspace{2mm}
\begin{minipage}{0.5\textwidth}
	\begin{algorithm}[H]
		\caption{MAKESET-IMPLICIT initializes the data structures}
		\footnotesize
		\begin{algorithmic}[1]
			\Function{MAKESET-IMPLICIT}{$x$}
			\If{$x$ not already present:}
			\State{add $x$ to the disjoint set tree}
			\State{$x.parent = x$}
			\State{$x.flag = {\sf TRUE}$}
			\State{$x.size = 1$}
			\vspace{1.5mm}
			\EndIf
			\EndFunction
		\end{algorithmic}
		\label{algo: MAKESET-IMPLICIT(x)}
	\end{algorithm}
	\vspace{0.08in}
\end{minipage}
\noindent
\textbf{Algorithm: }
We use the Union-Find algorithm with path compression and union by size. As is well known, during path compression the parent of a vertex might change. This is when its flag also needs to be updated. 
FIND($x$) procedure (see Algorithm~\ref{algo: FIND(x)}) returns the root of the component containing $x$ and the updated flag of $x$ (note that flag of $x$ represents color of $x$ because parent of $x$ is root after FIND$(x)$).
Inside FIND($x$) procedure, we make a recursive call to FIND($x.parent$) which returns the root of the component and updated flag of $x.parent$. If $x.parent$ is not in the same bipartition as the root in G (i.e $x.parent.flag = {\sf FALSE}$) and $x$ and $x.parent$ are in the same bipartition in G (i.e. $x.flag = {\sf TRUE}$), then $x$ is not in the same bipartition as the root in G. Its flag must be changed to ${\sf FALSE}$ while updating its parent as root. The table below considers all the possible combinations and it corresponds to the XNOR gate.
\begin{table}[h!]
	\label{table1}
	\footnotesize
	\begin{center}
		\begin{tabular}{l|c|r} 
			\textbf{$x.flag$} & \textbf{$x.parent.flag$} & \textbf{Updated $x.flag$}\\
			\hline
			& &\\
			${\sf TRUE}$ & ${\sf TRUE}$ & ${\sf TRUE}$\\
			${\sf TRUE}$ & ${\sf FALSE}$ & ${\sf FALSE}$\\
			${\sf FALSE}$ & ${\sf TRUE}$ & ${\sf FALSE}$\\
			${\sf FALSE}$ & ${\sf FALSE}$ & ${\sf TRUE}$\\
		\end{tabular}
	\end{center}
\end{table}
\begin{algorithm}[H]
	\footnotesize
	\caption{UNION-IMPLICIT handles insertion of edge $(x,y)$}
	\begin{algorithmic}[1]
		\Function{UNION-IMPLICIT}{$x,y$}
		\State{$(xRoot, xColor) = FIND(x)$}
		\State{$(yRoot, yColor) = FIND(y)$}
		\If{$xRoot == yRoot$} 
		\If{$xColor == yColor$}
		\State{ERROR, INVALID EDGE INSERTION}
		\Else
		\State{RETURN} 
		\EndIf
		\Else 
		\If{$xRoot.size < yRoot.size$}
		\State{swap($xRoot$, $yRoot$)} 
		\State{swap($xColor$, $yColor$)}
		\EndIf
		\State{$yRoot.parent = xRoot$}
		\State{$xRoot.size = xRoot.size + yRoot.size$}
		\State{$yRoot.flag = xColor$ XOR $yColor$}
		\EndIf
		\EndFunction
	\end{algorithmic}
	\label{algo: UNION-IMPLICIT(x,y)}
\end{algorithm}
\noindent
When an edge $(x,y)$ is inserted, we call UNION-IMPLICIT($x,y$) (Algorithm~\ref{algo: UNION-IMPLICIT(x,y)}). The procedure works as follows : We call FIND($x$) and FIND($y$). FIND($x$) returns root of $x$(denoted by $xRoot$) and color of $x$(denoted by $xColor$). Similarly, FIND($y$) returns root $yRoot$ and color $yColor$ of $y$. We have the following cases:  
\begin{enumerate}
	\item $xRoot$ is same as $yRoot$.
	\begin{enumerate}
		\item $xColor$ is same as $yColor$. Then this edge insertion introduces an odd length cycle violating the bipartiteness property of the graph. Hence, such edge insertions are invalid.
		\item $xColor$ is different from $yColor$. Then no further processing is required. 
	\end{enumerate} 
	\item $xRoot$ is different from $yRoot$. Then, the insert merges two components. To complete merging, without loss of generality, we assume that $x$ belongs to the larger component. We make $xRoot$ the parent of $yRoot$. If $x$ and $y$ have the same color, then the colors of all vertices in $y$'s component need to be flipped. But in \textit{Implicit} coloring we only need to change color for $yRoot$ and this is achieved by changing $yRoot.flag$ to ${\sf FALSE}$. This implicitly recolors all vertices in $y$'s component as their flags represent relative coloring. $yRoot.flag$ can be obtained by using the XOR operation on $xColor$ and $yColor$. Note that $yRoot$ is the only vertex whose parent has been updated and hence the only vertex whose flag needs to be updated.
\end{enumerate}
\begin{table}[h!]
	\vspace{-1mm}	
	\begin{center}
		\label{tab:table2}
		\begin{tabular}{l|c} 
			\textbf{Update/Query} & \textbf{Action}\\
			\hline
			Insert($x,y$) & Call UNION-IMPLICIT($x,y$)\\
			get-Color(x) & Call FIND($x$), Return $x.flag$
		\end{tabular}
	\end{center}
\end{table}
The extra running time incurred by the above procedures compared to the standard Union-Find is $O(1)$. Hence, the time complexity for Union-Find holds here. The amortized running time per operation is $O(\alpha(n))$, where $\alpha(n)$ is the inverse Ackermann function. $\alpha(n)$ is less than $5$ for any $n$ that can be expressed in the physical universe. Hence, each of the above operations run in near constant time. We get the following Theorem:
\begin{theorem}
\label{thm:IncrementalBipartiteImplicit}
Starting with an empty bipartite graph with $n$ vertices, an implicit 2-coloring can be maintained deterministically over any sequence of insertion in amortized $O(\alpha(n))$ update time, amortized $O(\alpha(n))$ query time, and amortized $O(\alpha(n))$ many recolorings, where $\alpha(n)$ is the inverse Ackermann function. 
\end{theorem}




\section{Fully dynamic $(\Delta + 1)$-coloring for general graphs}
\label{sec: GeneralGraph}
In this section, we present a deterministic fully dynamic algorithm for $\Delta + 1$ coloring with at most one  recoloring, worst case query time of $O(1)$, 
and worst case update time of $O(\sqrt{m})$. Recall that
 $\Delta$ is the maximum degree of a vertex in the graph and $m$ is the maximum number of edges in the graph at any point during the update sequence. Our algorithm
is inspired by the deterministic dynamic algorithm for maintaining a maximal matching by Neiman and Solomon~\cite{Neiman}; however we require a different set of data structures.\\
\noindent {\bf Overall approach:}
We partition the vertex set $V$ into two sets $V_{\HIGH}$ and $V_{\LOW}$. 
If  $deg(v) \geq \sqrt{2m}$ then $v \in V_{\HIGH}$; otherwise, $v \in V_{\LOW}$. 
Thus, $|V_{\HIGH}| \le \sqrt{2m}$ and  as in \cite{Neiman}, our goal is to have a bounded number of {\em high-degree} vertices. 
We say that  a color is \textit{used} for a vertex if it is assigned to at least one of its neighbors; otherwise we
say that the color is \textit{free} for the vertex. For every vertex $v \in V_{\HIGH}$, we maintain the set of used colors $\U(v)$ and free colors $\F(v)$. 
We aim to find a free color for a high-degree vertex efficiently. In contrast, for a low-degree vertex we can scan
all its neigbhors to find a free color.\\
\\{\bf Data Structures:}
We assume that we have  $\Delta+1$ colors numbered as $\{1,2,\ldots,\Delta+1\}$. 
We maintain the following data structures :
\begin{enumerate}
\item {\bf Color of a vertex:} For every vertex $v$, we store the color $C_v$ assigned to $v$. 
\item {\bf Degree of a vertex:} For every vertex $v$, we maintain its degree $d_v$ as an integer. This allows us to check if a vertex is of high degree or not in $O(1)$ time.
\item {\bf Neigbhorhood of a vertex:} For every vertex $v$, we maintain the set of edges incident on $v$ as
 a doubly linked list $\N(v)$. In addition, for every vertex we maintain an array($A_v$) of size $n$ which augments the list $\N(v)$. 
If the edge $(i, v)$ is present in $G$, then at index $i$ in the array, we store the pointer to the edge $(i, v)$ in $\N(v)$. Otherwise the index $i$ in the array points to {\sf  NULL}.  
\item {\bf Lists of Free and Used colors:} If $v \in V_{\HIGH}$, then we maintain two separate lists $\U(v)$ and $\F(v)$ as doubly linked lists.
The list $\U(v)$ contains the set of colors used by the neighbors of $v$, whereas the list $\F(v)$ contains the set of colors which are not used by the neighbors of $v$.
We remark that the color $C_v$ is contained in $\F(v)$.   
We augment these doubly linked lists with arrays $\FC_v$ and $\UC_v$ of size $\Delta+1$.\\
\begin{itemize}
\item  If color $j$ is in $\F(v)$ then entry $\FC_v[j]$, that is, the $j$-th entry of the free array contains a pointer to the node in the doubly linked list $\F(v)$ which stores the color $j$.
 The entry $\UC_v[j]$, that is, the $j$-th entry of  used array contains NULL. 
\item  If color $j$ is in $\U(v)$ then entry $\UC_v[j]$, that is, the $j$-th entry of the used array contains a pointer to the node in the doubly linked list $\U(v)$ which stores the color $j$.
 The entry $\FC_v[j]$, that is, the $j$-th entry of  free array contains NULL. 
\end{itemize}
It is easy to verify that using the above data structures we perform the following operations in constant time: (i) check whether a given color is free or used for a vertex $v$, (ii)
change the status of a given color for a vertex $v$ from used to free and vice versa.
\item {\bf Count of neighbours of each color:} For every high degree vertex $v$, we maintain an array $\CT_v$ of size $\Delta+1$. The $i$-th index of 
$\CT_v$ denotes the number of neighbors of $v$ which are colored using the color $i$.
\item {\bf List of high degree vertices: } We maintain the list of high degree vertices as a doubly linked list $L_{\HIGH}$. Clearly $|L_{\HIGH}| \leq \sqrt{2m}$.\\ 
\end{enumerate}
\begin{minipage}{0.5 \textwidth}
\textbf{Initialization: } Procedure {\sf INIT}(Algorithm~\ref{algo: Init}) is called to initialize the data structures required for each vertex. Even the data structures used only by $\textit{high-degree}$ vertices are initialized here although they will not be maintained until the vertex becomes $\textit{high-degree}$.\\
\end{minipage}
\hspace{2.5mm}
\begin{minipage}{0.5 \textwidth}
	\vspace{-1cm}
	\begin{algorithm}[H]
		\footnotesize
		\caption{{\sf INIT} initializes the data structures}
		\begin{algorithmic}[1]
			\Function{Init}{}
			\For{$v \in V$}
			\State $C_v =1$ ; $d_v = 0$ ; $\N(v) = NULL$
			\State $\U(v) = NULL$
			\For{$c \in \{1, \ldots, \Delta+1\}$}
			\State Add $c$ to $\F(v)$
			\EndFor
			\EndFor
			\State $L_{\HIGH} = NULL$
			\EndFunction
		\end{algorithmic}
		\label{algo: Init}
	\end{algorithm}
\end{minipage}
\begin{minipage}{0.5 \textwidth}
\vspace{-0.5cm}	
\begin{algorithm}[H]
	\footnotesize
	\caption{{\sf INSERT} handles insertion of edge $(u,v)$}
	\begin{algorithmic}[1]
		\Function{Insert}{$u,v$}
		\State Add $v$ to $\N(u)$ and $u$ to $\N(v)$
		\State Increment $d_u$ and $d_v$
		\If{$C_u == C_v$}
		\State RECOLOR($v$)
		\EndIf
		\If{$d_u == \sqrt{2m}}$
		\State Add $u$ to $L_{\HIGH}$
		\For{$x \in \N(u)$}
		\State $\CT_u[C_x]$ = $\CT_u[C_x]+1$
		\If{$\CT_u[C_x] == 1$}
		\State Move $C_x$ from $\F(u)$ to $\U(u)$
		\EndIf
		\EndFor
		\EndIf
		\If{$d_u \ge \sqrt{2m}}$
		\State $\CT_u[C_v]$ = $\CT_u[C_v]+1$
		\If{$\CT_u[C_v] == 1$}
		\State Move $C_v$ from $\F(u)$ to $\U(u)$
		\EndIf
		\EndIf
		\State Repeat steps 6-15 with respect to vertex $v$
		\EndFunction
	\end{algorithmic}
	\label{algo: Insert}
\end{algorithm}
\vspace{2mm}
\end{minipage}
\hspace{3mm}
\begin{minipage}{0.55 \textwidth}
\vspace{-0.3cm}	
\begin{algorithm}[H]
	\caption{{\sf RECOLOR} changes the color of vertex $v$}
	\footnotesize
	\begin{algorithmic}[1]
		\Function{RECOLOR}{$v$} 
		\If{$v \in L_{\HIGH}$}
		\State{Select a color $c$ from $\F(v)$}
		\Else
		\State{Scan $\N(v)$ and find a free color $c$}
		\EndIf
		\State{$C_{v\_old} = C_v$}
		\State{$C_v = c$}
		\For{$w \in L_{\HIGH}$} 
		\If{$\A_v[w] \neq NULL$}
		\State{$\CT_w[C_{v\_old}]$ = $\CT_w[C_{v\_old}]-1$}
		\If{$\CT_w[C_{v\_old}]=0$}
		\State{Move $C_{v\_old}$ from $\U(w)$ to $\F(w)$}
		\EndIf
		\State{$\CT_w[C_v]$ = $\CT_w[C_v]+1$}
		\If{$\CT_w[C_v]=1$}
		\State{Move $C_v$ from $\F(w)$ to $\U(w)$}
		\EndIf
		\EndIf
		\EndFor
		\EndFunction
	\end{algorithmic}
	\label{algo : Recolor(v)}
\end{algorithm}
\vspace{0.1mm}
\end{minipage}
\textbf{Edge Insertion:}
Procedure {\sf INSERT(u,v)} (Algorithm~\ref{algo: Insert}) is called when an edge between vertices $u$ and $v$ is inserted. If vertex $u$ has the same color as vertex $v$ the we call Procedure RECOLOR($v$) (Algorithm~\ref{algo : Recolor(v)}). Procedure RECOLOR($v$) works as follows: if vertex $v$ is of high degree then a color $c$ is selected from the list $F(v)$. Otherwise, we search neighbors of $v$ and find a free color $c$. We assign color $c$ to vertex $v$ and inform all its high degree neighbors about the change in color of $v$. Time taken by RECOLOR($v$) is $O(\sqrt{m})$. It is straightforward that {\sf INSERT} runs in $O(\sqrt{m})$ time.\\\\\\\\
\noindent
\textbf{Edge Deletion:}
Procedure {\sf DELETE(u,v)} (Algorithm~\ref{algo: Delete}) is called when an edge between vertices $u$ and $v$ is deleted. 
It is straightforward that {\sf DELETE(u,v)} runs in $O(\sqrt{m})$ time.
\begin{algorithm}[H]
	\footnotesize
	\caption{{\sf DELETE} handles deletion of edge $(u,v)$}
	\begin{algorithmic}[1]
		\Function{Delete}{$u,v$}
		\State Remove $u$ from $\N(v)$ and $v$ from $\N(u)$
		\State Decrement $d_u$ and $d_v$
		\If{$d_u == \sqrt{2m}-1$}
		\For{$c \in \U(u)$}
		\State Move $c$ from $\U(u)$ to $\F(u)$
		\EndFor
		\EndIf
		\If{$d_ u \ge \sqrt{2m}$}
		\State $\CT_u[C_v]$ = $\CT_u[C_v]-1$ 
		\If{$\CT_u[C_v] == 0$}
		\State Move $C_v$ from $\U(u)$ to $\F(u)$
		\EndIf
		\EndIf
		\State Repeat steps 4-10 with respect to vertex $v$
		\EndFunction
	\end{algorithmic}
	\label{algo: Delete}
\end{algorithm}
Therefore, updates take $O(\sqrt{m})$ time and queries take $O(1)$ time in the worst case.
\begin{theorem}
	\label{thm:FullyDynamicGeneralGraph}
	Starting with an empty graph with $n$ vertices, a $(\Delta+1)$-coloring, performing one recoloring in the worst case, can be maintained deterministically over any sequence of insertion and deletion in worst case $O(\sqrt{m})$ update time and worst case $O(1)$ query time.	 	
\end{theorem}


\subsection{Low arboricity graphs}
\label{sec:Low-Arboricity-Graphs}
In this section we present a fully dynamic algorithm for maintaining a $(\Delta+1)$-coloring in graphs with bounded arboricity.
Arboricity $\gamma$ of a graph $G(V,E)$ is defined as:
\begin{center}
	\footnotesize	
	$\displaystyle {\gamma = \max_{U \subseteq V}\ceil[\bigg]{ \frac{|E(U)|}{|U|-1}}}$
\end{center}
where $E(U)$ = $\{(u,v) \in E | u,v \in U\}$. A dynamic graph is said to have arboricity $\gamma$ iff the arboricity of the graph remains bounded by $\gamma$ after any update. A $\beta$-orientation of an undirected graph $G(V,E)$ is a directed graph $H(V,A)$ where $A$ contains the edges in $E$ with a direction assigned to it such that the out-degree of every vertex in $H$ is at most $\beta$.
It is  known due to Nash-Williams~\cite{Nash-Williams} that a graph $G(V,E)$ has arboricity $\gamma$ iff $E$ can be partitioned as $E_1, E_2, \dots,E_{\gamma}$ such that $(V,E_i)$ is a forest for all $1 \leq i \leq \gamma$. This implies that a graph with arboricity $\gamma$ has a $\gamma$-orientation.  
Brodal and Fagerberg~\cite{DBLP:conf/wads/BrodalF99} studied the problem of maintaining $\gamma$-orientation for a graph with arboricity $\gamma$ in the fully dynamic setting. We use the following result
from their paper ~\cite{DBLP:conf/wads/BrodalF99} for our fully dynamic coloring algorithm.
\begin{theorem}
	\cite{DBLP:conf/wads/BrodalF99}	
	There exists a fully dynamic algorithm $\mathcal{B}$ that maintains an $O(\gamma)$ orientation with an amortized update time of $O(\gamma + \log n)$ in a dynamic graph with $n$ vertices and arboricity $\gamma$.
\end{theorem}
\textbf{Overall approach:} Every vertex $v$ maintains a list of neighbors {\sf N(v)} and a list of outgoing neighbors {\sf OUT(v)} 
(it contains $w$ if there is an edge $(v,w)$ directed from $v$ to $w$). Along with this,
every vertex maintains a \textit{partially correct} list {\sf Free(v)} of free colors. 
The list is \textit{partially correct} because if a color $c$ is present in the list {\sf Free(v)} then $c$ is not used by any of its neighbors in {\sf N(v) $\setminus$ OUT(v)} 
but $c$ may be used by some neighbors in {\sf OUT(v)}. Therefore, the free color information is correct with respect to the neighbors in {\sf N(v) $\setminus$ OUT(v)} 
(whose size is unbounded) but may be incorrect with respect to the neighbors in {\sf OUT(v)} (whose size is at most $O(\gamma)$). 
During our algorithm, a vertex $v$ is responsible for notifying  its color $C_v$ to all the neighbors in {\sf OUT(v)}. 
Therefore, whenever we need to recolor a vertex $v$, we call procedure {\sf RECOLOR-LowArboricity} (Algorithm~\ref{algo: RECOLOR-LowArboricity}) which works as follows: for every vertex $w \in {\sf OUT(v)}$, if $C_w$ is present in {\sf Free(v)} then we remove $C_w$ from {\sf Free(v)}. We pick a color $c$ from {\sf Free(v)} and assign it to $v$. For every vertex $w \in {\sf OUT(v)}$, if color $c$ is present in {\sf Free(w)} then it is removed from {\sf Free(w)}. {\sf RECOLOR-LowArboricity} takes $O(\gamma)$ time in the worst case.\\
\textbf{Data Structures:} We maintain the following for every vertex $v$:
\begin{itemize}
	\item We store the color $C_v$ assigned to $v$.	
	\item List of neighbors {\sf N(v)} as doubly linked list augmented with an array $A_v$ of size $n$.
	\item List of out-going neighbors {\sf OUT(v)} as doubly linked list augmented with an array $O_v$ of size $n$.
	\item \textit{Partially correct} list of free colors {\sf Free(v)} augmented with an array ${\sf Free\_{A_v}}$ of size $\Delta + 1$. 
	\item An array ${\sf COUNT_v}$ of size $\Delta+1$. The $i$-{th} index of ${\sf COUNT_v}$ denotes the number of neighbors of $v$ in ${\sf N(v) \setminus OUT(v)}$ which are colored using the color $i$.\\ 
\end{itemize}
Insert, Delete, and Search operations in the above data structures takes $O(1)$ time. Our algorithm works as follows:\\
\textbf{Orientation:} We call procedure {\sf Update-Orientation} (Algorithm~\ref{algo: Update-Orientation}) after every update. We use algorithm $\mathcal{B}$ to maintain $\gamma$ orientation. Then we update our data structures {\sf OUT} and {\sf Free}. If there are total $t$ edge re-orientations this step takes $O(t)$ time. Therefore, during entire run of the algorithm, this procedure takes amortized $O(\gamma + \log n)$ time.
\begin{algorithm}[H]
	\footnotesize
	\caption{{\sf Update-Orientation} maintains orientation of edges after an update}
	\begin{algorithmic}[1]
		\Function{Update-Orientation}{$u,v$}
		\State Run Algorithm $\mathcal{B}$
		\If{$(u,v)$ was inserted and directed from $u$ to $v$}
		\State Insert $v$ to ${\sf OUT(u)}$
		\If {$C_u$ is present in ${\sf Free(v)}$}
		\State Remove $C_u$ from ${\sf Free(v)}$
		\EndIf 
		\State ${\sf COUNT_v[C_u] = COUNT_v[C_u] + 1}$
		\EndIf
		\If{$(u,v)$ was deleted and directed from $u$ to $v$}
		\State Remove $v$ from ${\sf OUT(u)}$
		\State ${\sf COUNT_v[C_u] = COUNT_v[C_u] - 1}$
		\If {${\sf COUNT_v[C_u] == 0}$}
		\State Insert ${\sf C_u}$ to ${\sf Free(v)}$
		\EndIf
		\EndIf
		\For {every edge $(x,y)$ that was re-oriented to $(y,x)$ by algorithm $\mathcal{B}$}
		\State Remove $y$ from {\sf OUT(x)} and Insert $x$ to {\sf OUT(y)}
		\State ${\sf COUNT_y[C_x] = COUNT_y[C_x] - 1}$
		\State ${\sf COUNT_x[C_y] = COUNT_x[C_y] + 1}$
		\If {${\sf COUNT_y[C_x] == 0}$}
		\State Insert ${\sf C_x}$ to ${\sf Free(y)}$
		\EndIf
		\If {${\sf COUNT_x[C_y] == 1}$}
		\State Remove ${\sf C_y}$ from ${\sf Free(x)}$
		\EndIf
		\EndFor
		\EndFunction
	\end{algorithmic}
	\label{algo: Update-Orientation}
\end{algorithm}
\noindent
\textbf{Edge-Insertion:} We call procedure {\sf INSERT-LowArboricity(u,v)}(Algorithm~\ref{algo: INSERT-LowArboricity}) when an edge between two vertices $u$ and $v$ is inserted. We call procedure {\sf Update-Orientation(u,v)}. If vertex $u$ has same color as vertex $v$ and the edge is directed from $u$ to $v$ then we call procedure {\sf RECOLOR-LowArboricity(v)}. Otherwise, if vertex $u$ has same color as vertex $v$ and the edge is directed from $v$ to $u$ then we call procedure {\sf RECOLOR-LowArboricity(u)}. Other than call to {\sf Update-Orientation(u,v)}, remaining steps in {\sf INSERT-LowArboricity(u,v)} takes $O(\gamma)$ time in the worst case.\\
\begin{minipage}{0.5 \textwidth}
\begin{algorithm}[H]
	\footnotesize
	\caption{{\sf INSERT-LowArboricity(u,v)} handles insertion of edge $(u,v)$}
	\begin{algorithmic}[1]
		\Function{Insert-LowArboricity}{u,v}
		\State {Insert $v$ to ${\sf N(u)}$ and $u$ to ${\sf N(v)}$}
		\State Run Algorithm {\sf Update-Orientation(u,v)}
		\If {${\sf C_u == C_v}$}
		\If{direction is from $u$ to $v$}
		\State {\sf RECOLOR-LowArboricity(v)}
		\Else
		\State {\sf RECOLOR-LowArboricity(u)}
		\EndIf
		\EndIf
		\EndFunction
	\end{algorithmic}
	\label{algo: INSERT-LowArboricity}
\end{algorithm}
\vspace{-0.7cm}
\begin{algorithm}[H]
	\footnotesize
	\caption{{\sf DELETE-LowArboricity(u,v)} handles deletion of edge $(u,v)$}
	\begin{algorithmic}[1]
		\Function{DELETE-LowArboricity}{u,v}
		\State {Delete $v$ to ${\sf N(u)}$ and $u$ to ${\sf N(v)}$}
		\State Run Algorithm {\sf Update-Orientation(u,v)}
		\EndFunction
	\end{algorithmic}
	\label{algo: DELETE-LowArboricity}
\end{algorithm}
\vspace{0.8mm}
\end{minipage}
\hspace{2mm}
\begin{minipage}{0.5\textwidth}
\vspace{-0.2cm}	
\begin{algorithm}[H]
	\footnotesize
	\caption{{\sf RECOLOR-LowArboricity(v)} changes the color of vertex $v$}
	\begin{algorithmic}[1]
		\Function{RECOLOR-LowArboricity}{v}
		\For{${\sf w \in OUT(v)}$}
		\If{$\sf C_w \in Free(v)$}
		\State Remove $\sf C_w$ from $\sf Free(v)$
		\EndIf
		\EndFor
		\State Select a color $c$ from {\sf Free(v)}
		\State ${\sf C_{v-old} = C_v}$
		\State $C_v = c$
		\For{${\sf w \in OUT(v)}$}
		\State ${\sf COUNT_w[C_{v-old}] =  COUNT_w[C_{v-old}] - 1}$
		\If {${\sf COUNT_w[C_{v-old}] == 0}$}
		\State Insert ${\sf C_{v-old}}$ to {\sf Free(w)}
		\EndIf
		\State ${\sf COUNT_w[C_{v}] =  COUNT_w[C_{v}] + 1}$
		\If {${\sf COUNT_w[C_{v}] == 1}$}
		\State Remove ${\sf C_{v}}$ from {\sf Free(w)}
		\EndIf		
		\EndFor	  
		\EndFunction
	\end{algorithmic}
	\label{algo: RECOLOR-LowArboricity}
\end{algorithm}	
\end{minipage}
\textbf{Edge-Deletion:} We call procedure {\sf DELETE-LowArboricity(u,v)} (Algorithm~\ref{algo: DELETE-LowArboricity}) when an edge between two vertices $u$ and $v$ is deleted. We call {\sf Update-Orientation(u,v)}. Other than call to {\sf Update-Orientation(u,v)}, remaining steps in {\sf DELETE-LowArboricity(u,v)} takes $O(1)$ time in the worst case.	
\begin{theorem}
\label{thm:FullyDynamicLowArboricity}
Starting with an empty graph with $n$ vertices with arboricity bounded by $\gamma$, a $(\Delta+1)$-coloring, performing one recoloring in the worst case, can be maintained deterministically over any sequence of insertion and deletion in amortized $O(\gamma + \log n)$ update time and worst case $O(1)$ query time.  
\end{theorem}

\bibliography{Reference}
\appendix

\end{document}